\documentclass[journal,10pt]{IEEEtran}
\ifCLASSINFOpdf
   \usepackage[pdftex]{graphicx}
  
\else

\fi

\usepackage{epsfig}
\usepackage{epstopdf}
\usepackage{amsmath,amssymb}
\usepackage{epstopdf}
\usepackage{latexsym}
\usepackage{epsfig}
\usepackage{pgfplots}
\usepackage[linesnumbered,ruled]{algorithm2e}
\usepackage{algpseudocode}
\usepackage{subcaption}
\usepackage{float}
\usepackage{amsfonts}
\usepackage{amsthm}
\usepackage{url}
\usepackage{bm}
\usepackage{textcomp}
\usepackage{dsfont}
\usepackage{footnote,footmisc}
\usepackage{enumitem}
\usepackage{caption}
\usepackage{array}
\newtheorem{theorem}{Theorem}
\newtheorem{mydef}{Definition}

\newtheorem{lemma}{Lemma}

\begin{document}

\title{Resource Allocation for Loss Tolerant Video Streaming in eMBMS}

\author{Sadaf~ul~Zuhra\IEEEauthorrefmark{1},~\IEEEmembership{Student Member,~IEEE,}
        Prasanna~Chaporkar\IEEEauthorrefmark{1},~\IEEEmembership{Member,~IEEE,}
        and~Abhay~Karandikar\IEEEauthorrefmark{1},~\IEEEmembership{Member,~IEEE}
\thanks{\IEEEauthorrefmark{1} The authors are with the Department of Electrical Engineering, Indian Institute of Technology Bombay.
e-mail: \{sadafshah, chaporkar, karandi\}@ee.iitb.ac.in.
A. Karandikar is currently the Director, Indian Institute of Technology Kanpur (on leave from IIT Bombay). e-mail:karandi@iitk.ac.in.}
 \vspace{-1cm}
}%

\maketitle

\begin{abstract}
Bandwidth hungry video content has become the dominant contributor to the data traffic world over. Cellular networks are constantly evolving to meet the growing traffic demands. Over the past few years, wireless multicast has been garnering a lot of attention as a means of efficient resource utilization. Multicast transmission lets spectral resources to be shared between users streaming the same content. Even though multicast transmission allows to serve multiple users on the same resources, in order to serve all these users successfully, the base station cannot transmit the content at a rate greater than that decodable by the user with the worst channel conditions. In this paper, we propose a way to overcome this bottleneck. Video streaming services can sustain a certain amount of packet loss without any significant degradation in the quality experienced by the users. We leverage this loss tolerant nature of video streaming applications to improve the performance of multicast video services in LTE and 5G. We convert the problem of resource allocation for loss tolerant multicasting into the problem of stabilizing a queueing system. We then propose two throughput optimal Maximum Weight (MW) policies that successfully stabilize the constructed queueing system. However, brute force implementation of MW policies is mostly NP-hard. To overcome this, we propose a maximum weight bipartite matching approach that results in a polynomial time implementation of the proposed policies. We also evaluate the performance of our policies via extensive simulations.  
\vspace{-0.5cm}
\end{abstract}

\section{Introduction} \label{sec:intro}
Video streaming is expected  to form a staggering $75 \%$ of the total mobile data traffic by the year 2021~\cite{ericsson}. With the huge amount of video content being generated online, videos are consuming a major chunk of the available network bandwidth. There is, therefore, a pressing need for resource allocation techniques that can result in efficient bandwidth utilization and reduce the burden on the network resources. Video streaming applications have been found to be tolerant of packet losses as high as $40 \%$~\cite{h264}. Currently, H.264/AVC is the most commonly used video codec over the Internet. For an H.264 encoded video, decoders like FFmpeg and JM can successfully conceal as much as $39 \%$ packet loss with no deterioration in the quality of the video observed by the end user~\cite{h264}. This loss tolerant nature of video streaming applications can be leveraged to build video specific resource allocation schemes which can reduce the congestion in the network significantly. Tolerating some losses in video streams can allow for higher transmission rates and greater flexibility in providing video streaming services. However, this loss tolerant nature of video streaming has not been exploited in the existing multicast streaming literature. \par

Abundant availability of online streaming services also leads to a lot of content being redundantly streamed to several users simultaneously. This means that a large number of resources could be in use for transmitting the same content to different users. Using multicast transmission, the number of resources needed for a single user can be used for catering to all the users streaming the same content. Multicast services in Long Term Evolution (LTE) are termed as evolved Multimedia Broadcast Multicast Services (eMBMS)~\cite{mbms_3gpp}. In eMBMS, the User Equipments (UEs) can subscribe to any of the multicast services available in the network. The UEs subscribed to a multicast service form a single multicast group. Once the eMBMS content is available, the evolved NodeB (eNB) informs the UEs of the service start and begins transmitting the content. The subscribed UEs can join the eMBMS session any time during the service. The eMBMS content is transmitted to the entire multicast group on the same Physical Resource Blocks (PRBs). Therefore, for every UE to receive the eMBMS content successfully, the eNB cannot transmit data to a group at a rate greater than that supported by the UE with the worst channel quality (henceforth referred to as the weakest UE) in the group. The weakest UE, therefore, becomes a bottleneck for the entire multicast group. \par

In this work, we take advantage of the loss tolerant nature of video streams to enhance the performance of video multicast streaming services. We design efficient resource allocation algorithms for loss tolerant video streaming and evaluate their performance through extensive simulations. Allowing for some controlled losses in a multicast stream gives us the flexibility of not having to serve every UE in the multicast group in every sub-frame. As long as we can keep the loss of a UE below a certain threshold, the UE will receive the desired video quality. This means that the weakest UE is no longer a bottleneck all of the time. The transmission rates in some sub-frames can, therefore, be higher than what can be decoded by the weakest UE in the group. This leads to higher throughput and better user satisfaction.

\vspace{-0.3cm}
\subsection{Related Literature}
To the best of our knowledge, none of the existing works deal with resource allocation for loss tolerant video streaming. Therefore, the related literature summarized in this section is primarily composed of the papers that deal with resource allocation for multicast transmission in general. The problem of grouping and resource allocation for multicast streaming has also been extensively studied by us in~\cite{tech_report} and~\cite{self}. In~\cite{tech_report} and~\cite{self}, we have addressed the problems of grouping and optimal resource allocation for multicast streaming in eMBMS. The objective of resource allocation in~\cite{tech_report,self} was to satisfy all the multicast UEs while minimizing the number of PRBs used in doing so. In~\cite{fdps}, the authors propose a Frequency Domain Packet Scheduler (FDPS) for Multimedia Broadcast Multicast Services (MBMS) that maximizes the minimum rate achievable by UEs in a PRB. It uses a somewhat pessimistic approach in that it only minimizes the damage caused by the worst PRB assignment. Moreover, the performance of the proposed policy has only been compared to a blind FDPS policy that uses a blind static allocation that doesn't change over time. This is not a good benchmark to compare with. In~\cite{mung_chiang}, the authors propose a fair and optimal resource allocation for eMBMS. It is assumed that the video content is simultaneously available through unicast as well as through eMBMS and the primary problem seeks to jointly optimize over the grouping of UEs and allocation of resources to unicast and eMBMS. The resource allocation scheme proposed in the paper allocates resources to groups proportional to the number of UEs in the group. While allocating resources, the varying channel conditions of UEs over different PRBs have not been considered.  \par

Resource allocation for MBMS Operation On-Demand has been studied in~\cite{mood}. The authors consider Quality of Experience (QoE) metrics such as user engagement instead of Quality of Service (QoS) metrics like average throughput to be the utility function sought to be maximized by the corresponding resource allocation schemes. All the video streams are assumed to be encoded using Scalable Video Coding (SVC). In~\cite{dash}, the authors have used convex optimization to obtain an optimal solution for multicasting Dynamic Adaptive Streaming over HTTP (DASH)~\cite{3gpp_dash} as well as SVC streaming content over LTE. The problem optimizes the Modulation and Coding Scheme (MCS) and the Forward Error Correction (FEC) code rates used while allocating resources. \par

In~\cite{pricing}, the authors use a pricing based scheme for allocating resources to multicast groups streaming SVC video content. Users are divided into three multicast groups based on the price they pay. UEs that pay the most receive the maximum number of enhancement layers.
In~\cite{tassi}, the authors investigate the use of Random Network Linear Coding (RNLC) for improving the performance of multicast services. They use two different forms of RNLC for multicasting H.264/SVC videos in a generic cellular system. The authors in~\cite{tassi2} deal with optimizing the delivery of network coded SVC content using eMBMS. They make use of Unequal Error Protection (UEP) for ensuring reliability of the multi layer video transmission. They propose a UEP Resource Allocation Model (UEP-RAM) that seeks to maximize the profit to cost ratio of the system. It is shown that the proposed UEP-RAM provides a much better coverage than conventional multi-rate transmission~\cite{afolabi}. Even though SVC provides an interesting new method of video encoding with various benefits, H.264/AVC continues to be the choice of encoding videos over the Internet. Most of the popular streaming platforms like Netflix~\cite{netflix} and YouTube~\cite{youtube} use H.264/AVC or VP9 to encode their videos. Therefore, in this work, we do not consider a layered video coding for the eMBMS video streams.
\par

In~\cite{nest}, the authors propose a scheduling scheme for eMBMS broadcast services that is focused on reducing the average latency of packets in the system. The proposed scheme starts transmission in unicast mode and gradually moves to the broadcast mode as the number of UEs increases. In~\cite{bolte}, the authors deal with efficient broadcasting in LTE using eMBMS. The proposed broadcasting mechanism has been given the name of Broadcast over LTE (BoLTE). Their resource allocation algorithm uses a water filling form of the proportional fair scheduling~\cite{waterfill1,waterfill2}. The authors have evaluated the performance of BoLTE using a WiMAX testbed. \par

The existing literature does not take advantage of the fact that since video streaming is the most useful application of eMBMS,
the inherent loss tolerance of video streams can be leveraged for improving the performance of multicast services. In this paper, we consider the use of eMBMS primarily for video streaming applications. Since video streams can handle losses as high as $40 \%$ without significantly altering the quality perceived by the end-user~\cite{h264}, we exploit this property to design efficient resource allocation policies for eMBMS. A similar approach has also been used in~\cite{atilla} to design a regular service guarantee algorithm for a wireless network with a number of links, only one of which can transmit in a time slot. In our system, the degree of packet loss that can be tolerated by an end user depends upon factors like the video stream subscribed and the channel quality experienced by the user. Allowing for some losses can help in better resource utilization and in controlling congestion in the network during peak traffic hours. It also reduces the dependence of a multicast group on the UE with the worst channel quality as the resource allocation policy is no longer constrained to serve every UE in every sub-frame. \par
In most of the existing multicast literature, the rate achievable by a group is assumed to be the same over all PRBs. This assumption significantly simplifies the resource allocation problem. Without the channel variability over PRBs, all PRBs are equivalent for a UE and the problem reduces to determining the number of PRBs allocated to a multicast group. We, however take into account the fact that due to fading, the channel states and hence the Channel Quality Indicator (CQI) values for a group or user may vary over different PRBs in a single sub-frame as well. Therefore, it is not enough to determine the number of PRBs to be allocated to a group, the identity of the PRBs being allocated also needs to be specified. 
\vspace{-0.5cm}

\subsection{Contributions}
The main contributions of this paper are summarized below: \par
$\bullet$ We address the problem of resource allocation in loss tolerant eMBMS systems in which every user may have a different loss tolerance. This is the first work of its kind that considers a loss tolerant model for video multicast services. \\
$\bullet$ We convert the problem of resource allocation in loss tolerant eMBMS networks to the problem of stabilizing a queueing system. We prove that stabilizing the token queues in the constructed queueing system is equivalent to satisfying the loss requirements of the users. \\
$\bullet$ We propose an optimal online Maximum Weight (MW) resource allocation policy for loss tolerant eMBMS networks. The policy doesn't require any statistical information of the channel states of users. Channel states can vary arbitrarily and can also be correlated across users. The proposed policy is optimal in the sense that it can stabilize the queueing system whenever any other policy, including offline policies with complete information of channel states of users, can do so.\\
$\bullet$ We propose another online priority Maximum Weight (MW-priority) resource allocation policy. It provides an improvement over the MW policy in terms of the burstiness of the losses encountered by the users. MW-priority ensures that no user is starved for long periods at a stretch, thus providing a better quality of experience. MW-priority is also optimal in the sense that it can stabilize the queueing system whenever any other policy, including offline policies with complete information of the channel states, can do so. \\
 $\bullet$ We also present a mechanism for a polynomial time implementation of the proposed policies (which are otherwise computationally very expensive and impractical to implement) using Maximum Weight Bipartite Matching (MWBM). \\
$\bullet$ We have evaluated the performance of the proposed schemes via extensive simulations. We have compared their performance to that of the throughput optimal Exponential (Queue length) (EXP-Q) rule~\cite{shakkottai}. EXP-Q does not take the loss requirements of users into consideration. Under this policy, several users encounter losses greater than their thresholds. On the other hand, our proposed policies  successfully meet the loss requirements of all the users. Thus, by taking the loss tolerance of users into consideration, we are able to satisfy a larger number of users. \par

The rest of this paper is organized as follows. We discuss the system model and the problem formulation in Sections~\ref{sec:system} and~\ref{sec:problem} respectively. The construction of the queueing system and related results are presented in Section~\ref{sec:queueing}. In Section~\ref{sec:allocation}, we present the proposed resource allocation algorithms and in Section~\ref{sec:mwbm}, we discuss the MWBM based polynomial time implementation of these algorithms. The details of the simulations done are given in Section~\ref{sec:sim} and the conclusions in Section~\ref{sec:conclusion}. In the interest of preserving the flow of the paper, proofs of all the theorems and lemmas are given separately in the appendix.

\section{System Model} \label{sec:system}
Our system consists of an LTE cell with $L$ different eMBMS services available for transmission to the UEs. There are $M$ UEs in the cell that can subscribe to any of these eMBMS services. Let $[n] = \{1,\ldots,n\}$ and let $|A|$ denote the cardinality of a set $A$. Thus, $[M]$ and $[L]$ denote the set of UEs and the set of multicast groups, respectively.
UEs subscribed to the $i^{th}$ video stream form multicast group $G_i$. The number of UEs in $G_i$ is denoted by $K_i$. Each eMBMS group is allocated one PRB in every sub-frame. A resource allocation policy $\Gamma$ decides which PRB will be allocated to which group in every sub-frame.
We define an allocation vector ${\bf B^\Gamma[t]}$ for policy $\Gamma$ in sub-frame $t$. ${\bf B^\Gamma[t]}$ is a vector of length $L$ that specifies which PRB, if any, has been assigned to each group. Note that $\Gamma$ is completely defined by the value of ${\bf B^\Gamma[t]}$ in every sub-frame $t$. We use $B^\Gamma_i[t]$ to denote the $i^{th}$ entry of vector ${\bf B^\Gamma[t]}$. If $G_i$ is not scheduled for reception in sub-frame $t$, then $B^\Gamma_i[t] = 0$, otherwise $B^\Gamma_i[t]$ takes the value of the PRB number allocated to $G_i$.
The $i^{th}$ eMBMS service requires data to be transmitted to the subscribed UEs at rate $R_i$. For each eMBMS stream, a data packet arrives at the beginning of every sub-frame and is transmitted in the same sub-frame.

\par

The channel states of UEs vary across time and frequency. As a result, the channel experienced by a UE varies from one sub-frame to another and also across PRBs in a sub-frame. Depending on the CQI experienced by UE $k$ in PRB $j$ in sub-frame $t$, there is a certain maximum MCS that an be supported in that PRB for that UE~\cite{cqi} and a corresponding maximum rate that the UE can successfully decode. We denote this rate as $r_{kj}[t]$. Whenever a PRB is allocated to multicast group $G_i$, data is transmitted in that PRB at the corresponding rate $R_i$. Since every UE in the multicast group experiences a different channel quality in the allocated PRB, the maximum rate that UEs can successfully decode is also different for each of them. As a result, a UE in a multicast group may not receive the transmitted content successfully even after a PRB has been assigned to its group. When a UE successfully receives data in a sub-frame, we say that the UE has been served in that sub-frame. Note that a UE being \textit{scheduled} and being \textit{served} is not the same. We distinguish between these two terms below:\\
$\bullet$ We say that a UE has been scheduled in a sub-frame if a PRB is allocated to its corresponding group in that sub-frame. For instance, UE $k \in G_i$ is said to have been scheduled for reception in sub-frame $t$ under policy $\Gamma$ if $B^\Gamma_i[t] \neq 0$. \\
$\bullet$ We say that a UE has been served in a sub-frame if it has been scheduled in that sub-frame and is able to successfully decode the received content. For instance, UE $k \in G_i$ is said to have been served in sub-frame $t$ under $\Gamma$ if $B^\Gamma_i[t] = j \neq 0$ and $R_i \leq r_{kj}[t]$. \par
We denote the loss encountered by UE $k$ under policy $\Gamma$ in sub-frame $t$ by $\ell_k^\Gamma[t]$. For UE $k \in G_i$ and $B^\Gamma_i[t] = j \neq 0$, we have:
  \begin{eqnarray} \label{eq:loss_def}
 \ell^\Gamma_{k}[t] = \begin{cases}
	{0}, & \text{if} \ R_i \leq r_{kj}[t], \\
	{1}, &  \text{otherwise}.
\end{cases}
\end{eqnarray}
For $B^\Gamma_i[t] = 0$, UE $k$ is not scheduled for reception and so, $\ell^\Gamma_{k}[t] = 0$.
Since we assume that a packet arrives at the beginning of every sub-frame and is transmitted in the same sub-frame, the fraction of sub-frames in which a UE is not served under policy $\Gamma$ is equal to the fractional packet loss for that UE under $\Gamma$. Thus, $\ell^\Gamma_{k}[t]$ represents the packet loss encountered by UE $k$ under policy $\Gamma$.

 \par 

As discussed in the previous section, video streams are loss tolerant to a certain extent. The loss tolerance is different for different UEs since UEs experience different channel conditions and are subscribed to different services. The loss tolerance of a UE depends on the channel conditions experienced and the video resolution chosen by it. Higher resolutions mean a lower loss tolerance and vice versa. A UE that chooses automatic resolution adjustment will have the highest loss tolerance. We use $\tilde{\ell}_k$ to denote the fractional loss that can be tolerated by UE $k$.\ $\pmb{\tilde{\ell}} = [\tilde{\ell}_1,\ldots,\tilde{\ell}_M]$ is the loss tolerance vector for the system. In the next section, we formally define the resource allocation problem.
\vspace{-0.2cm}
\section{Problem Definition} \label{sec:problem}
We begin by stating some important definitions that will be used in defining the problem statement.

\begin{mydef} \label{def:RA} \textit{Feasible resource allocation}:
Resource allocation in a sub-frame is said to be feasible if it assigns at most one PRB to each multicast group such that no two groups are assigned the same PRB. In other words,
a feasible resource allocation in sub-frame $t$ corresponds to an allocation vector ${\bf B^\Gamma[t]}$ such that no two non-zero elements in it are equal i.e., if $B^\Gamma_i[t] \neq 0$, then $B^\Gamma_i[t] \neq B^\Gamma_{i'}[t]$ for every $i' \neq i$. 
\end{mydef}

\begin{mydef} \label{def:RAP} \textit{Feasible resource allocation policy}:
A feasible resource allocation policy $\Gamma$ is a policy that chooses a feasible allocation vector in every sub-frame.
\end{mydef}

A resource allocation policy can make use of the knowledge of current channel states of UEs, the allocation information of the previous sub-frames, the loss tolerance of UEs and
the losses encountered by the UEs in the past to make allocation decisions in a sub-frame. It could even be an off-line policy that could make the allocation decisions in advance if the channel conditions of all sub-frames are known apriori.

\begin{mydef} \textit{Average Packet loss}:
 We denote the average packet loss encountered by a UE $k$ under resource allocation $\Gamma$ by $\bar{\ell}_{k}^\Gamma$. 
 It is the total packet loss per unit time and can be mathematically written as follows:
 \begin{equation*}
  \bar{\ell}_{k}^\Gamma = \lim \sup_{T \rightarrow \infty} \frac{1}{T} \sum_{t=1}^T \ell^\Gamma_{k}[t].
 \end{equation*}
\end{mydef}

\begin{mydef} \textit{Feasible region of a policy}:
The feasible region of a resource allocation policy $\Gamma$, ${\cal L}^\Gamma$, is the set of all loss tolerance vectors, $\pmb{\tilde{\ell}}$s that can be satisfied by $\Gamma$ i.e. $\pmb{\tilde{\ell}} > \pmb{\bar{\ell}^\Gamma}$ with probability (w.p.) $1$.
\end{mydef}

\begin{mydef} \textit{Feasible region of the system}:
 The feasible region of the system is the set of loss vectors ${\cal L} = \bigcup_\Gamma{\cal L}^\Gamma$ where the union is over all feasible $\Gamma$.
\end{mydef}

\begin{mydef} \textit{Optimal policy}:
 The optimal resource allocation policy $\Gamma^\star$ is a policy whose feasible region is the set of loss vectors ${\cal L}^{\Gamma^\star} = \bigcup_\Gamma{\cal L}^\Gamma$.
\end{mydef}

Our objective here is to determine the optimal resource allocation policy $\Gamma^\star$. We design the optimal resource allocation policy using results from queueing theory. Towards that end, we convert the resource allocation problem in a loss tolerant eMBMS network to the problem of stabilizing a queueing system and prove that stabilizing the resulting system is equivalent to meeting the loss requirements of the UEs.

\vspace{-0.3cm}
\section{Queueing System for Resource Allocation in Loss Tolerant $e$MBMS} \label{sec:queueing}
We convert the problem of resource allocation in a loss tolerant eMBMS network to the problem of obtaining a throughput optimal allocation policy for stabilizing a queueing system. Towards this end, we first discuss the construction of the queueing system. 
\vspace{-0.5cm}
\subsection{Construction}
The queueing system consists of token queues corresponding to every UE. At the beginning of a sub-frame, a token arrives in the token queue of UE $k$ w.p. $(1-\tilde{\ell}_k)$ and w.p. $\tilde{\ell}_k$, there is no new arrival. We represent the token arrival process for queue $k$ by an indicator random variable $\lambda_k[t]$. $\lambda_k[t]$ is $1$ if a token arrives in queue $k$ in sub-frame $t$ and $0$ otherwise:
\begin{eqnarray*}
\lambda_{k}[t] = \begin{cases}
{1}, & \text{w.p.} \ 1-\tilde{\ell}_k \\
{0}, & \text{w.p.} \ \tilde{\ell}_k.
\end{cases}
\end{eqnarray*}
The arrivals are independent and identically distributed, therefore, the average rate of token arrival to queue $k$ is $\lambda_k = (1-\tilde{\ell}_k)$. We use $\pmb{ \lambda} = \{\lambda_1,\ldots,\lambda_M\}$ to denote the arrival rate vector of the system.
$Q_{k}[t]$ denotes the length of queue $k$ at the beginning of sub-frame $t$. If UE $k$ is successfully served in a sub-frame, a token departs from its token queue and the queue length is reduced by $1$. 
 \par

 We define another indicator random variable $\mu^\Gamma_k[t]$ that indicates whether or not UE $k$ has been served in sub-frame $t$ under $\Gamma$. $\mu^\Gamma_k[t] =1$ iff $k$ is served under $\Gamma$ in sub-frame $t$. Say $k \in G_i$ and $B_i^\Gamma[t] = j$. Then, $\mu^\Gamma_k[t] = 1$ iff $j \neq 0$ and $R_i \leq r_{kj}[t]$. Otherwise, $\mu^\Gamma_k[t]=0$.
$\pmb{\mu^\Gamma[t]} = [\mu^\Gamma_1[t],\ldots,\mu^\Gamma_M[t]]$ denotes the service vector of the system in sub-frame $t$ under policy $\Gamma$. The stability region of the queueing system thus constructed can be defined as follows:

\begin{mydef}
 Stability region of the queueing system: The queueing system is said to be stable if the expected queue lengths stay finite for every queue i.e. $\sup_t \mathbb{E}[Q_k[t]] < \infty$ for every $k$.
 A resource allocation policy that stabilizes the system is called a stable resource allocation policy. The stability region of a resource allocation policy $\Gamma$ is the set of arrival rate vectors for which the system is stable under $\Gamma$.
 The stability region of the queueing system is the union of the stability regions of all feasible $\Gamma$. We denote it as ${\cal S}$.
 \end{mydef}

\begin{mydef}
Throughput optimality: 
A resource allocation policy $\Gamma$ is said to be throughput optimal if $\Gamma$ can stabilize the queueing system if any other policy can do so. This means that if the queueing system is at all stabilizable, $\Gamma$ will succeed in stabilizing it.
\end{mydef}

The queueing system thus constructed can be maintained at the eNB. Since the eNB knows the loss requirements of the UEs as well as their channel states, it has all the information needed to maintain the queueing system. In the next section, we examine the stability region of the constructed queueing system and relate it to the feasible region of the optimal resource allocation policy.
\vspace{-0.3cm}
\subsection{Feasible Region of the eMBMS System and Stability Region of the Queueing System} \label{subsec:dtmc}
In this section, we prove that stabilizing the queueing system constructed in the previous section is equivalent to meeting the loss requirements of the UEs. This will establish the equivalence of the stability region of the constructed queueing system and the feasible region of the optimal resource allocation policy. We begin by defining a few terms. \par
Define a set ${\cal B} = \{B_1, \ldots, B_{|{\cal B}|}\}$ containing all possible PRB allocation vectors to the $L$ groups. The cardinality of this set $|{\cal B}| =\left(_L^N\right) \times L!$. In LTE, channel states are quantified in terms of CQI values. According to Third Generation Partnership Project (3GPP) standards~\cite{cqi}, a total of $15$ CQI values are defined in LTE. Since the number of CQI values are finite, the possible channel states of UEs can take finitely many values. We define a set ${\cal C}$ that contains all possible channel state combinations of all the UEs in the system. For an LTE system, ${\cal C}$ will therefore be a set of $15^M$ CQI vectors, each of size $M$. Let $g$ be the probability distribution over the set ${\cal C}$. That is, the channel state of the system in a sub-frame $t$, $C(t) = C$ w.p. $g(C)$.
We denote by $\pmb{\mu_{B_iC}}$, the vector of service rates of UEs corresponding to an allocation $B_i$ in CQI state $C \in {\cal C}$. Note that $\pmb{\mu_{B_iC}}$s are binary vectors of size $M$. Define a distribution ${\bf w} = \{w_{B_iC}\}$ over the set of $\pmb{\mu_{B_iC}}$s where $w_{B_iC}$ denotes the probability of choosing allocation $B_i$ in channel state $C \in {\cal C}$. Using these definitions, we define the the following LP:
\begin{align*}
LP(\delta): 
 &\hspace{0.5cm} \sum_{C \in {\cal C}} \sum_{B_i \in {\cal B}} g(C) w_{B_iC} \pmb{\mu_{B_iC}} = \pmb{ \lambda}+\delta, \\
 &\hspace{0.5cm} w_{B_iC} \geq 0 \ \forall \ B_i \in {\cal B}, \ C \in {\cal C}, \\
 &\hspace{0.5cm} \sum_{B_i \in {\cal B}} w_{B_iC} = 1, \ \forall \ C \in {\cal C},
\end{align*}
where $\delta$ is a non-negative real number. Denote by $\Lambda(\delta)$ the set of arrival rate vectors $\pmb{\lambda}$ such that the feasible region of $LP(\delta)$ is non-empty.
Define two sets, $\Lambda^\circ = \bigcup_{\delta>0} \Lambda(\delta)$ and $\overline{\Lambda} = \bigcup_{\delta \geq 0} \Lambda(\delta)$.
In the next result, we establish the relation between sets $\Lambda^\circ$, $\overline{\Lambda}$ and stability region of the queueing system ${\cal S}$. This result is essential for relating the feasible region of the optimal resource allocation policy to the stability region of the queueing system.

\begin{theorem} \label{theorem:relation}
 $\Lambda^\circ \subseteq {\cal S} \subseteq \overline{\Lambda}$.
\end{theorem}

\begin{proof}
 The detailed proof is given in Appendix~\ref{proof:relation}.
\end{proof}
From this point forward, we consider $\Lambda^\circ$ to be the stability region of the queueing system. We now state and prove the following important theorem that relates the feasible region of the optimal resource allocation policy to the stability region of the queueing system.

\begin{theorem} \label{lemma:stable}
The loss requirement of a UE is met iff its token queue in the queueing system is stable. Therefore, the feasible region of the optimal allocation policy $\Gamma^\star$, ${\cal L}^{\Gamma^\star}$ is equivalent to the stability region of the queueing system, $\cal{S}$. i.e. $\pmb{\tilde{\ell}} \in {\cal L}^{\Gamma^\star}$ iff $(\pmb{1}-\pmb{\tilde{\ell}}) \in \cal{S}$. Here, $\pmb{1}$ is a vector of ones of same size as $\pmb{\tilde{\ell}}$.
\end{theorem}
\begin{proof}
The detailed proof is given in Appendix~\ref{proof:stable}.
\end{proof}



We have now established that the stability region of the constructed queueing system is same as the feasible region of the optimal resource allocation policy $\Gamma^\star$. Therefore, from this point onwards, we do not explicitly talk about meeting the loss
requirements of the UEs. Instead, we focus our attention on stabilizing the token queues corresponding to each UE knowing that stabilizing the token queues of UEs will ensure that their respective loss requirements are met. In the next section, we propose maximum weight throughput optimal policies for resource allocation in loss tolerant eMBMS systems.

\section{Proposed Resource Allocation Algorithms} \label{sec:allocation}
 In this section, we propose Maximum Weight (MW) throughput optimal policies for resource allocation in loss tolerant eMBMS networks. We also present their efficient polynomial time implementations in a later section.
\vspace{-0.4cm}
\subsection{Maximum Weight Resource Allocation $(\Gamma_0)$} \label{subsec:BP}
Maximum weight resource allocation $\Gamma_0$ takes scheduling decisions in a sub-frame $t$ based on the token queue lengths $Q_k[t]$s. In sub-frame $t$, $\Gamma_0$ chooses service vector $\pmb{\mu^{\Gamma_0}[t]}$ according to the following optimization problem:

\begin{equation} \label{eqn:mw_objective}
 \pmb{\mu^{\Gamma_0}[t]} = \arg\max_{\mkern-40mu \mu^{\Gamma_0}_{k}[t] \in \pmb{\mu_{B_iC}}}  \sum_{k=1}^{M} Q_{k}[t] \mu^{\Gamma_0}_{k}[t],
\end{equation}


where $\mu^{\Gamma_0}_{k}[t]$ is the service rate of UE k in sub-frame $t$ under policy $\Gamma_0$.
$\Gamma_0$ seeks to maximize the sum of the queue lengths of the UEs that are being served in sub-frame $t$. It has a computational complexity of $\mathcal{O}(M\binom{N}{L} L!)$.
We have already established in Section~\ref{sec:queueing} that stabilizing the token queues ensures that the loss requirements of the UEs are met. Therefore, to prove that $\Gamma_0$ can successfully meet the loss requirements of the multicast UEs, it is sufficient to show that $\Gamma_0$ stabilizes the constructed queueing system. We prove this in the following result.

\begin{theorem} \label{theorem:BP}
For any stabilizable arrival rate vector $\pmb{ \lambda}$, $\Gamma_0$ stabilizes the queueing system.
\end{theorem}
This theorem implies that as long as the system is stabilizable, i.e. there exists some policy $\Gamma$ that can stabilize the queueing system, so can $\Gamma_0$. Note that the policy $\Gamma$ is not restricted to use the same information that is available to $\Gamma_0$. $\Gamma$ could be using information of the past and future allocations and channel conditions to take allocation decisions. Despite that, we claim that $\Gamma_0$ will successfully stabilize the system using only the knowledge of the current state of the 
queueing system to make the scheduling decisions.
\begin{proof}
The detailed proof is given in Appendix~\ref{proof:BP}.
\end{proof}
We now have a MW policy that takes allocation decisions based on the UE token queue lengths and meets the loss requirements of users.
However, in addition to the amount of packet loss encountered by a video stream, the pattern in which the loss occurs is also an important factor in determining the quality experienced by the end user. While a $30 \%$ loss spread over multiple frames may not result in any significant quality impairment, a bursty loss of $30 \%$ of the packets at once can degrade the quality significantly for a while. Such bursts of degraded quality are unacceptable to the end users and may result in UEs choosing to leave the eMBMS session. It is therefore important to ensure that the packet losses are distributed throughout the multicast session. \par

In order to ensure this, a resource allocation policy also needs to restrict the amount of consecutive packet losses encountered by a UE in addition to the long term average packet loss. We propose such a policy which we call as the priority MW policy, in the next section. This policy improves upon MW policy by increasing the scheduling probability of a UE every time it is left unserved. This will ensure that a UE does not remain unserved for long periods which leads to better loss performance and reduces the burstiness of the losses encountered.

\subsection{Priority Maximum Weight Resource Allocation $(\Gamma_P)$} \label{subsec:BP_priority}
Priority Maximum Weight (MW-priority) resource allocation $\Gamma_P$ also takes scheduling decisions in a sub-frame based on the queue lengths $Q_k[t]$s in that sub-frame. However, in this scheme,
we use an additional priority vector to increase the probability of serving a previously unserved queue.
In every sub-frame $t$, $\Gamma_P$ chooses service vector $\pmb{\mu^{\Gamma_P}[t]}$ according the following optimization problem:
\begin{equation}
 \label{eq:objective_priority} 
 \pmb{\mu^{\Gamma_P}[t]} = \arg\max_{\mkern-40mu \mu^{\Gamma_P}_{k}[t] \in \pmb{\mu_{B_iC}}} \sum_{k=1}^M \left(Q_k[t] +(c_k[t]+1)\times s \right)\mu^{\Gamma_P}_k[t],
\end{equation}

where $\mu^{\Gamma_P}_{k}[t]$ is the service rate of UE k in sub-frame $t$ under $\Gamma_P$, $c_k[t]$ is the priority weight ascribed to the token queue of UE $k$ and $s$ is a positive constant. $c_k[t]$ is defined as follows:
\begin{equation*}
 c_k[t] = \begin{cases}
          {0}, & \text{if} \ \mu_k[t-1] = 1,\\
	{\min(c_k[t-1]+1,\kappa)}, &  \text{otherwise}.
         \end{cases}
\end{equation*}
$\kappa$ is the maximum positive integer value that the priority weights can take. Also, $c_k[0] = 0, \forall \ k$. We use ${\bf c}[t] = [c_1[t],\ldots,c_M[t]]$ to denote the vector of priority weights of all the queues in sub-frame $t$. 
$\Gamma_P$ also has a computational complexity of $\mathcal{O}(M\binom{N}{L} L!)$.
Since increasing $c_k[t]$ increases the contribution of UE $k$ in~(\ref{eq:objective_priority}), it is more likely to be served by the resource allocation policy. \par

When using policy $\Gamma_P$ for resource allocation, the state of the queueing system can be completely defined by the queue lengths of all the token queues in the system and the value of the priority counter for each queue. As before, we denote the state in sub-frame $t$ under policy $\Gamma_P$ by the vector ${\bf Q^{\Gamma_P}[t]} = [Q^{\Gamma_P}_1[t],\ldots,Q^{\Gamma_P}_M[t],\bar{c}[t]]$. Since the scheduling decisions under $\Gamma_P$ in a sub-frame are based only on the state of the system in that sub-frame, the evolution of states of the system form a Discrete Time Markov Chain (DTMC). In the next result we prove that this DTMC is countable, irreducible and aperiodic.

\begin{lemma} \label{lemma:dtmc2}
 The DTMC formed by the evolution of the states under $\Gamma_P$ ${\bf Q^{\Gamma_P}[t]} = [Q^{\Gamma_P}_1[t],\ldots,Q^{\Gamma_P}_M[t],\bar{c}[t]]$ is countable, irreducible and aperiodic.
\end{lemma}
\begin{proof}
 The detailed proof is given in Appendix~\ref{proof:dtmc2}.
\end{proof}

We now prove that $\Gamma_P$ is a throughput optimal policy i.e. $\Gamma_P$ will stabilize the queueing system if any other policy can do so.

\begin{theorem} \label{theorem:BP2}
 For any stabilizable arrival rate vector $\pmb{ \lambda}$, $\Gamma_P$ stabilizes the queueing system.
\end{theorem}

Theorem~\ref{theorem:BP2} states that if the queueing system under consideration is at all stabilizable, $\Gamma_P$ will stabilize it.

\begin{proof}
The detailed proof is given in Appendix~\ref{proof:BP2}.
\end{proof}
Thus, $\Gamma_P$ and $\Gamma_0$ are both throughput optimal.
In the next section, we discuss the Exponential (Queue length) rule (EXP-Q) which was proposed in~\cite{shakkottai}. We use the EXP-Q rule as a benchmark for performance evaluation of our policies since it is a well known throughput optimal policy for scheduling multiple flows over a time varying wireless channel. The EXP-Q rule also minimizes the maximum delay encountered in the system~\cite{shakkottai2004}. The rule, however, considers that there is a single channel that can be used by one flow at a time. In the next section, we present a generalization of EXP-Q for use with multicast transmission and with multiple channels available for allocation. We have modified the rule so it can be used for scheduling in the system under consideration in this paper. 
\vspace{-0.3cm}
\subsection{Generalized Exponential (Queue length) rule $(\Gamma_E)$} \label{subsec:exp_q}
EXP-Q rule is a throughput optimal policy~\cite{shakkottai} that schedules a single queue $k$ in a time slot $t$ such that:
\begin{equation} \label{eq:expq_objective}
 k \in \arg\max_{\mkern-25mu k} \gamma_k \mu_k[t] \exp \left( \frac{a_k Q_k[t]}{\beta + [\bar{Q}[t]]^\eta} \right),
\end{equation}
where $\mu_k[t]$ is the rate of service of queue $k$ in sub-frame $t$, $a_k$s, $\gamma_k$s and $\eta$ are constants and $\bar{Q}[t] = (1/N) \sum_k a_kQ_k[t]$. The EXP-Q rule assumes that there is a single time varying channel that is to be shared by several flows. We generalize the EXP-Q rule for use in LTE eMBMS where multiple channels (in the form of PRBs) are available for scheduling the $L$ multicast flows in the cell. Each channel is time varying and different for each UE in the cell. Moreover, in this case, we need to schedule groups of UEs in a single PRB instead of individual UEs since each multicast flow corresponds to a group of UEs subscribed to it. Hence, the EXP-Q rule cannot be used as it is for scheduling in eMBMS. \par

We use $\Gamma_E$ to refer to the modified EXP-Q rule. Since we have multiple channels available and multiple groups can be scheduled for service in a time slot, the policy has to choose an allocation vector instead of a single queue to be scheduled in a time slot.
Specifically, for the problem being considered in this paper, the policy needs to determine the allocation vector ${\bf B^{\Gamma_E}[t]}$ (as defined in Section~\ref{sec:problem}).
We define $\Gamma_E$ as the policy that chooses allocation vector $\pmb{\mu^{\Gamma_E}[t]}$ according to the following optimization problem:
\begin{equation*}
\pmb{\mu^{\Gamma_E}[t]} = \arg\max_{\mkern-40mu \mu^{\Gamma_E}_{k}[t] \in \pmb{\mu_{B_iC}}} \sum_{k=1}^M \gamma_k \mu^{\Gamma_E}_k[t] \exp \left( \frac{a_k Q_k[t]}{\beta + [\bar{Q}[t]]^\eta} \right),
\end{equation*}
where $\mu^{\Gamma_E}_{k}[t]$ is the service rate of UE k in sub-frame $t$ under $\Gamma_E$. $\Gamma_E$ also has a computational complexity of $\mathcal{O}(M\binom{N}{L} L!)$.
The allocation vector ${\bf B^{\Gamma_E}[t]}$ determines which PRB is allocated to which multicast group and hence determines the $\mu^{\Gamma_E}_k[t]$s.
The data is then transmitted to the $i^{th}$ group at the corresponding rate $R_i$ in the PRB allocated to it.
\vspace{-0.3cm}
\subsection{Computational Complexity} \label{subsec:complexity}
As mentioned before, the resource allocation policies discussed in this section have a computational complexity of $\mathcal{O}(M\binom{N}{L} L!)$. These are, therefore, not suitable for use in practical systems unless we can find efficient means of implementing them. It turns out that these policies can infact be implemented in polynomial time using a Maximum Weight Bipartite Matching (MWBM)~\cite{clr}. We discuss the details of this implementation in the next section.
 We first present the algorithm for implementing $\Gamma_0$ using MWBM in detail. The same algorithm can be used for implementing $\Gamma_P$ as well as $\Gamma_E$ by replacing the edge weights of $\Gamma_0$ with those of $\Gamma_P$ and $\Gamma_E$ respectively.

\section{Maximum Weight Bipartite Matching for $e$MBMS Resource Allocation} \label{sec:mwbm}
We make use of MWBM for an efficient polynomial time implementation of the resource allocation policies proposed in Section~\ref{sec:allocation}. The MWBM brings down the computational complexity of their implementation to $\mathcal{O}(NL^2)$. The policies can thus be implemented in polynomial time. We begin with the construction of the underlying bipartite graph which is the same for all the policies except for the edge weights which change according to the policy under consideration. We discuss the implementation for $\Gamma_0$ in detail. The procedure and proof involved can be directly used for $\Gamma_P$ and $\Gamma_E$ as well with modified edge weights. The modifications involved will be specified at
the end of this section. \par

Construct a bipartite graph ${\cal G} = (U,V,E)$ where vertex set $U$ is the set of $L$ multicast groups and vertex set $V$ is the set of $N$ PRBs. We define the service rate of a UE $k \in G_i$ in PRB $j$ in sub-frame $t$ as follows:
\begin{eqnarray*}
\mu_{k}^j[t] = \begin{cases}
{0}, & \text{if} \ R_i > r_{kj}[t] \\
{1}, & \text{otherwise}.
\end{cases}
\end{eqnarray*}
The weight of an edge connecting vertex $i \in U$ to vertex $j \in V$, $w_i^j[t]$ is the sum of the products of the queue lengths of UEs in group $G_i$ and their achievable service rates in PRB $j$ in sub-frame $t$ i.e. 
\begin{equation*}
w_i^j[t] = \sum_{k \in G_i} Q_k[t]\mu_k^j[t].
\end{equation*}
The resulting bipartite graph is illustrated in Figure~\ref{fig:graph}.

\begin{figure}[h]
	\centering
	\scalebox{0.7}{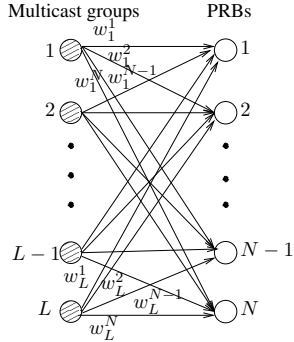}
	\caption{Bipartite graph between multicast groups and PRBs}
	\label{fig:graph}
	\end{figure}
	
A MWBM of ${\cal G}$ that matches every node in $U$ to a unique node from $V$ results in an allocation equivalent to $\Gamma_0$. We prove this in the following result.

\begin{lemma} \label{lemma:mwbm}
 Maximum weight bipartite matching for graph ${\cal G}$ results in resource allocation according to policy $\Gamma_0$.
\end{lemma}

\begin{proof}
The detailed proof is given in Appendix~\ref{proof:mwbm}.
\end{proof}
 
  The same MWBM can be used for implementing $\Gamma_P$ and $\Gamma_E$ by changing the edge weights. For $\Gamma_P$ we will have:
 \begin{equation} \label{eq:weight_priority}
  w_i^j[t] = \sum_{k \in G_i} \left(Q_k[t] +(c_k[t]+1)\times s \right)\mu_k^j[t].
 \end{equation}
For $\Gamma_E$ the edge weights will change to:
 \begin{equation} \label{eq:weight_expq}
  w_i^j[t] = \sum_{k \in G_i} \gamma_k \mu_k^j[t] \exp \left( \frac{a_k Q_k[t]}{\beta + [\bar{Q}[t]]^\eta} \right).
 \end{equation}
 The same proof as in Lemma~\ref{lemma:mwbm} follows to show that the MWBM for graph ${\cal G}$ with the edge weights defined in (\ref{eq:weight_priority}) and (\ref{eq:weight_expq})
 results in the implementation of resource allocation policies $\Gamma_P$ and $\Gamma_E$ respectively.
 In the next section, we present the results of the simulations performed for evaluating the performance of the proposed resource allocation schemes.

    \begin{center}
\begin{table} \vspace{-0.5cm}
\captionof{table}{System Simulation parameters \cite{parameters}} \label{parameters} 
\begin{center}
\begin{tabular}{ | m{3 cm} | m{5 cm}|} 
 \hline
 \textbf{Parameters} & \textbf{Values}\\
 \hline
 \hline
 System bandwidth & $20$ MHz \\
 \hline
 eNB cell radius & $150$ m \\ 
 \hline
 Path loss model & L = $128.1 + 37.6 \log10(d)$, $d$ in kilometers\\ 
 \hline
 Lognormal shadowing & Log Normal Fading with $10$ dB standard deviation \\
 \hline
 White noise power density & $-174$ dBm/Hz \\
 \hline
 eNB noise figure & $5$ dB \\
 \hline
 eNB transmit power & $46$ dBm \\
 \hline
 Number of PRBs & $100$ per sub-frame\\
\hline
 \end{tabular}
 \end{center}
 \end{table}
 \vspace{-0.5cm}
\end{center}
 \vspace{-0.4cm}

\section{Simulations} \label{sec:sim}
For studying the performance of the proposed allocation algorithms, we use these algorithms for allocation in LTE eMBMS. We consider an LTE cell with $L$ different eMBMS video streams and UEs distributed uniformly at random through the cell. All UEs are subscribed to one of the eMBMS video streams. UEs subscribed to a particular eMBMS service receive the relevant content on common PRBs. We have used a MATLAB~\cite{matlab} based LTE simulator designed in~\cite{mahima}. In order to create LTE specific physical layer conditions, we have 
   created channels using the models recommended by 3GPP in~\cite{parameters}. The SNR to CQI and CQI to rate mapping has been done using the tables specified in 3GPP documents~\cite{parameters}. Other relevant simulation parameters are given in Table~\ref{parameters}.\par
   
   Every multicast service has a certain rate requirement and every UE can tolerate a certain amount of packet loss. The loss tolerable by a UE depends on the quality of video required by it, what its average channel conditions are like and the multicast group it belongs to. PRBs are allocated to the multicast groups according to the MW, MW-priority as well as the modified EXP-Q resource allocation policies. We then observe the packet loss encountered by the UEs under each of these policies. We compare the performance of the proposed schemes with the modified EXP-Q rule~\cite{shakkottai} which is proven to be throughput optimal~\cite{shakkottai}. The simulations have been carried out for $10^6$ sub-frames and losses have been calculated as the fraction of sub-frames in which a UE is not served. \par

    \begin{figure*}[htb]
    \vspace{-0.6cm}
    \centering
    \begin{subfigure}[htb]{0.31\textwidth}
        \centering
\includegraphics[scale = 0.28]{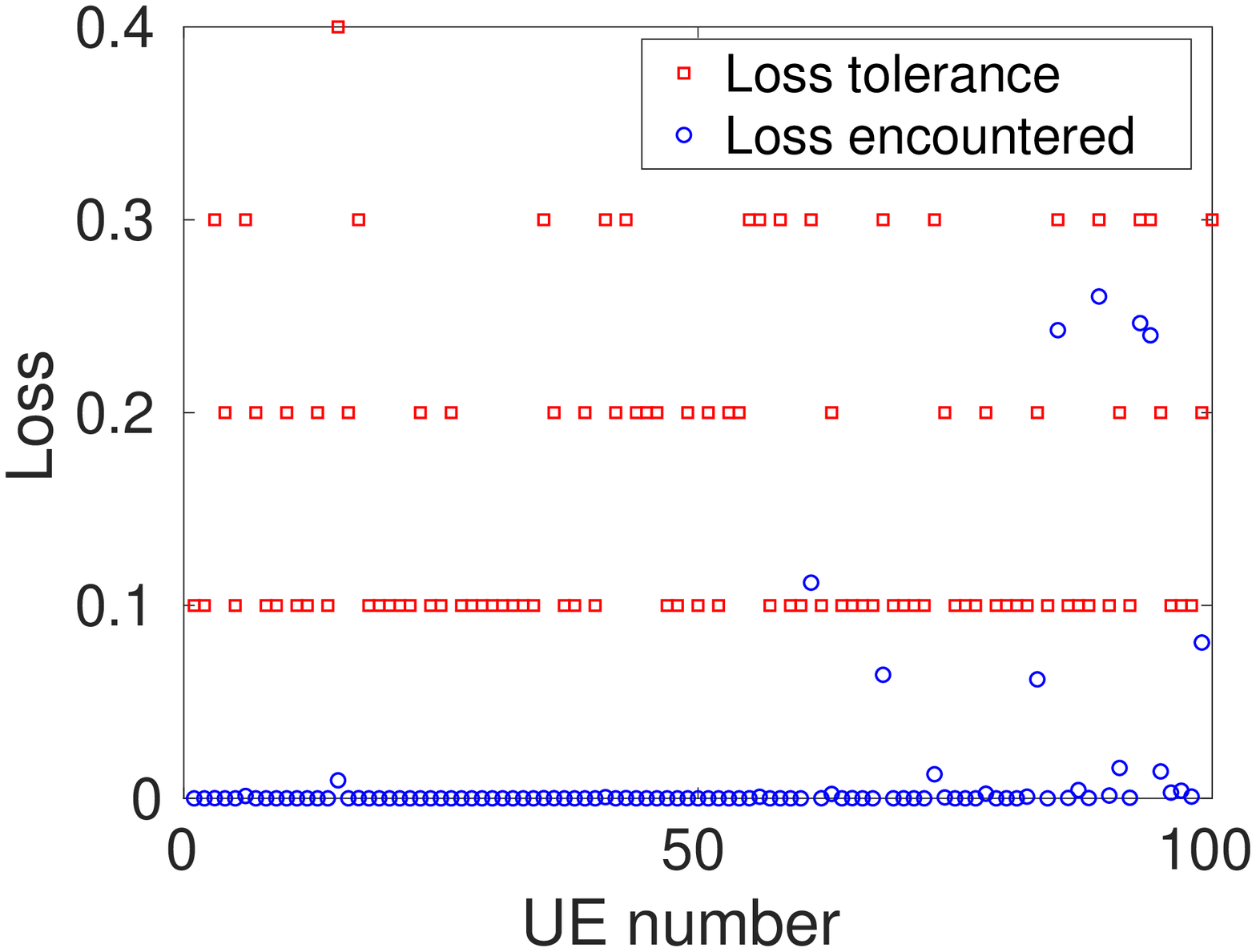} 
\caption{}
\label{loss_mw}
    \end{subfigure}%
    ~ 
    \begin{subfigure}[htb]{0.31\textwidth}
        \centering
\includegraphics[scale = 0.28]{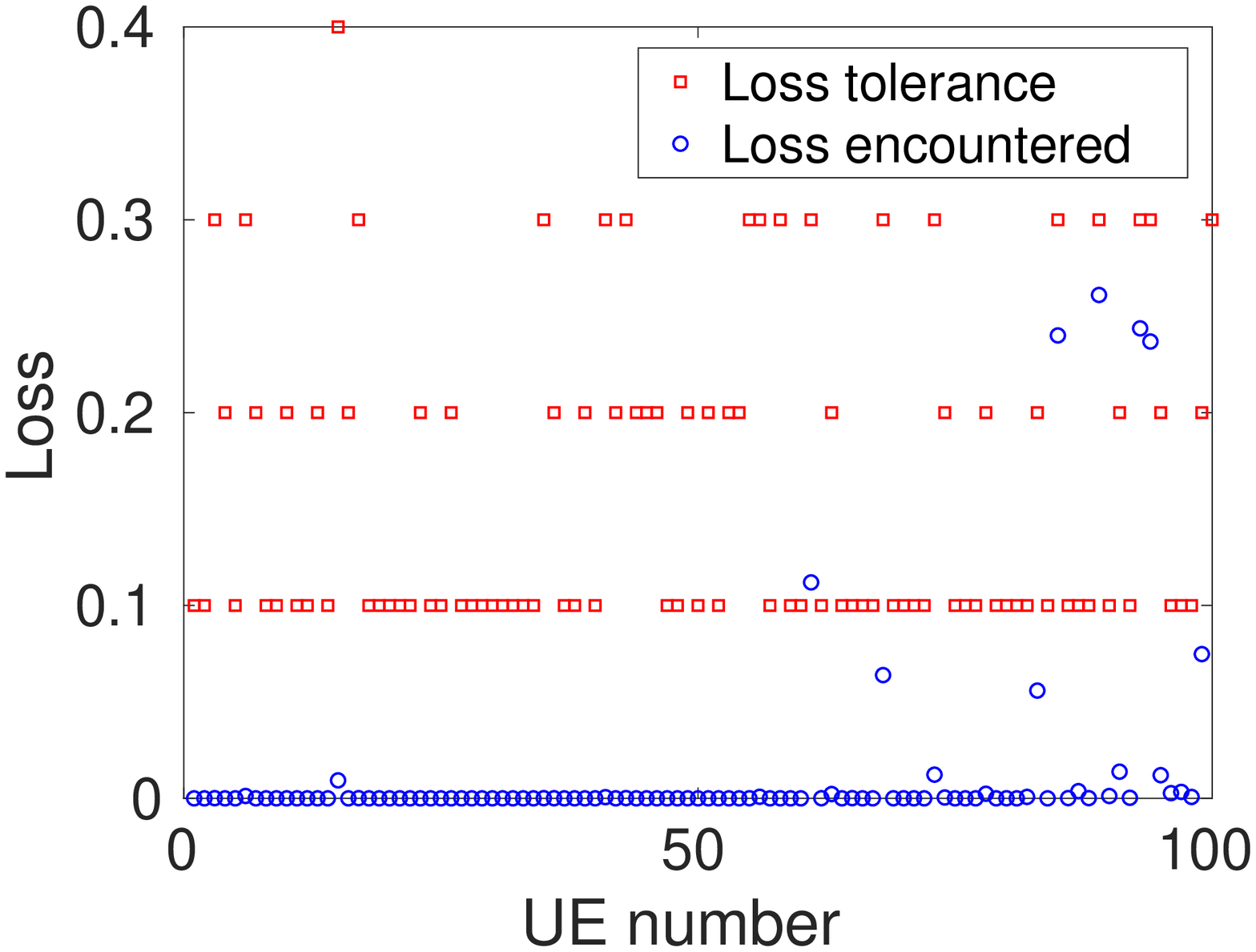}
\caption{}
\label{loss_mwp}
    \end{subfigure}%
~
    \begin{subfigure}[htb]{0.31\textwidth}
        \centering
\includegraphics[scale = 0.28]{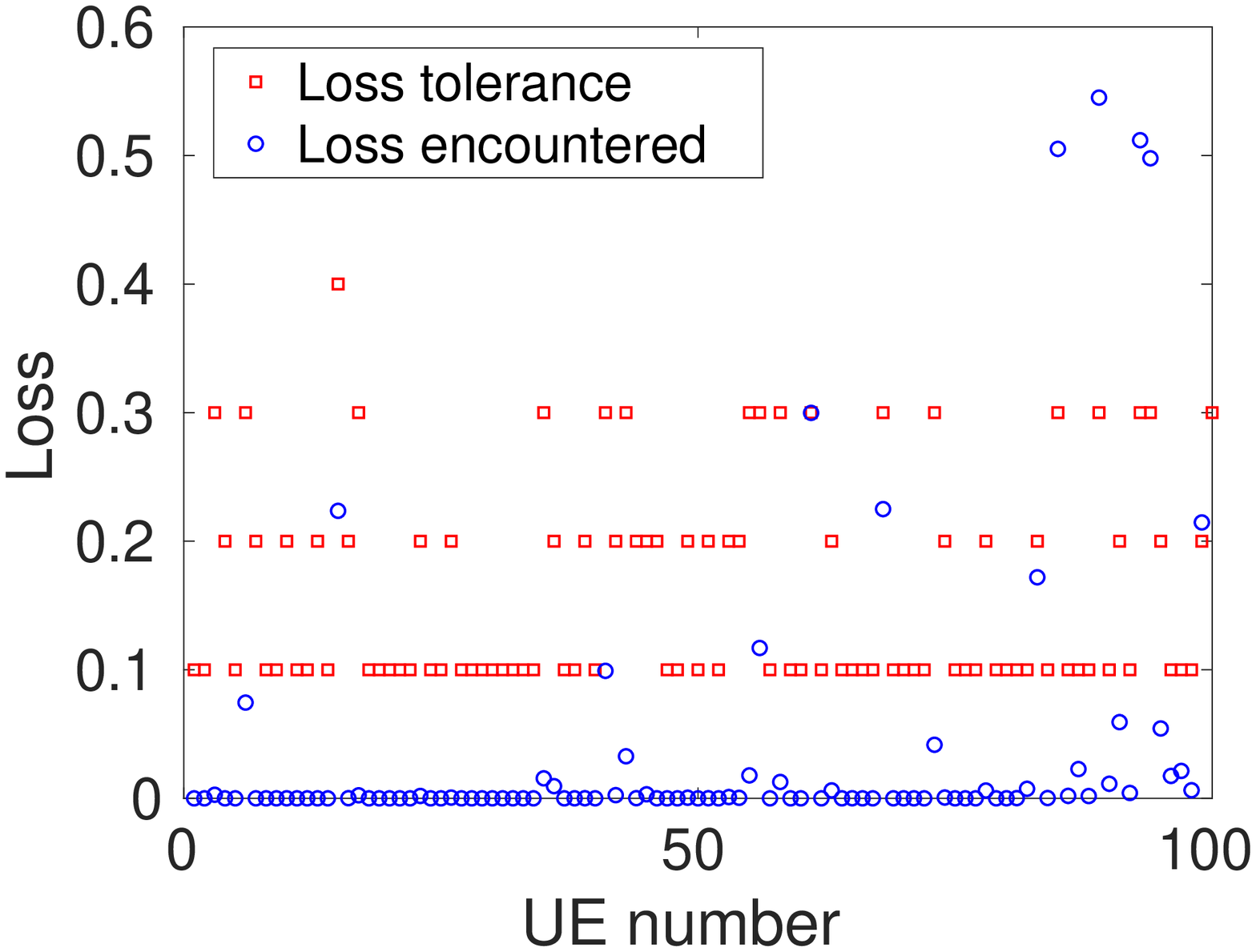}
\caption{}
\label{loss_exp}
    \end{subfigure}%
~
\caption{Tolerable loss versus loss encountered using a)\ MW b)\ MW-priority c)\ EXP-Q}
\label{fig1}
\vspace{-0.3cm}
\end{figure*}

    \begin{figure}[htb]
    \centering
    \includegraphics[scale = 0.3]{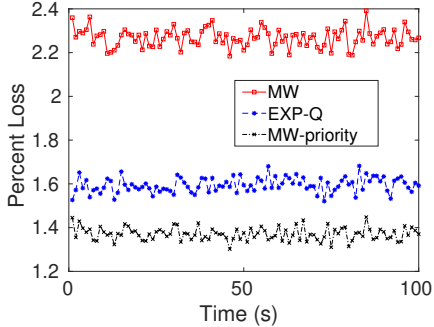} 
\caption{Average losses in MW, EXP-Q and MW-priority schemes}
\label{loss_avg}
\vspace{-0.4cm}
    \end{figure}%
    
    \begin{figure}[htb]
\includegraphics[scale = 0.4]{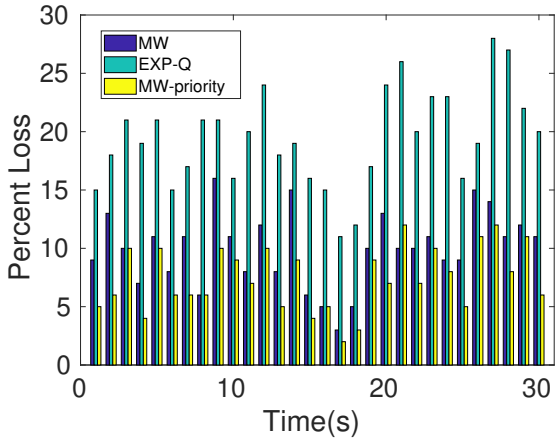}
\caption{Loss pattern of a UE (losses per sub-frame)}
\label{loss_bar}
\vspace{-0.3cm}
    \end{figure}%

   The first set of plots in Figure~\ref{fig1} compare the losses encountered by UEs to their respective loss tolerances. Figures~\ref{loss_mw} and~\ref{loss_mwp} show this comparison for MW and MW-priority policies respectively. Both these policies succeed in meeting the loss requirements of all the UEs in the system. The queueing system is, thus, stable under both the proposed MW policies. 
   Figure~\ref{loss_exp} plots the losses encountered when using the modified EXP-Q rule for resource allocation. We observe that several UEs experience losses significantly greater than their tolerable limits and the queueing system is rendered unstable. \par
   
   Figure~\ref{loss_avg} compares the plots of the average losses encountered by the UEs under the three schemes. For this, the losses encountered per second have been exponentially averaged for a user. Every point in the plot is then obtained by averaging over all the UEs. We observe that EXP-Q rule results in a better loss performance than the MW scheme. Even though EXP-Q gives us a better average loss performance, as we saw in Figure~\ref{fig1}, it failed to meet the loss requirements of several UEs. On the other hand, despite a greater average system packet loss, MW is able to meet the loss requirements of all the UEs. The MW-priority scheme leads to the least average packet losses among the three schemes. \par
   
   As discussed in Section~\ref{sec:allocation}, in addition to the amount of packet loss seen by an end user, the pattern in which the losses occur also has a major bearing on the video quality experienced by the user. While some amount of packet loss spread more or less uniformly through a session may lead to no degradation in quality at all, a concentrated packet loss can be extremely annoying in a video stream and may even lead to UEs quitting the session altogether. In order to observe the pattern of packet loss encountered under the three policies over time, we plot the percent packet loss encountered by one of the users with a high loss threshold. This is plotted for all the policies as a function of time in Figure~\ref{loss_bar}. EXP-Q sees the most variable loss pattern. The losses per second see jumps as high as $10 \%$ from one second to another. The MW policy does better than EXP-Q. However, the MW-priority policy provides the most uniform loss pattern of the three. It, therefore, controls the burstiness of the losses encountered and ensures that no UE is starved for long periods at a stretch. \par
   
   We have verified the effectiveness of the proposed policies via simulations. We now present our conclusions.
   \vspace{-0.2cm}
\section{Conclusions} \label{sec:conclusion}
\vspace{-0.1cm}
In this paper, we have addressed the problem of resource allocation in loss tolerant eMBMS networks. eMBMS flows are primarily video streams that can tolerate a certain amount of packet loss without affecting the quality perceived by the end users. Each flow is allocated one PRB every sub-frame. We have converted the problem of determining the optimal resource allocation policy for the system under consideration to the problem of stabilizing a queueing system. We have proposed a MW and a MW-priority algorithm for allocating PRBs in loss tolerant multicast video streaming. Since MW schemes are typically computationally expensive to implement, we propose a MWBM that provides a polynomial time implementation of the proposed policies. We have also generalized the EXP-Q rule~\cite{shakkottai} for use with multiple channels and multicast transmission. Simulations have been performed to study and compare the loss performance of the MW, MW-priority and the modified EXP-Q policies. We have observed from the simulations that the MW-priority policy results in the least packet loss of all the policies. Usage of such a policy for video streaming using eMBMS can significantly lower the burden on the network and reduce congestion when multiple video streams are being transmitted simultaneously in a cell.

\vspace{-0.2cm}
\appendix
\vspace{-0.1cm}
\subsection{Proof of Theorem~\ref{theorem:relation}} \label{proof:relation}
We begin by constructing the following randomized resource allocation policy $\Gamma_\delta$ based on $LP(\delta)$ defined in Section~\ref{subsec:dtmc}:
\begin{mydef} \label{def:randomized}
Randomized allocation policy $\Gamma_\delta$: $\Gamma_\delta$ chooses an allocation vector in a sub-frame according to a feasible solution ${\bf w}$ of $LP(\delta)$. If the system is in channel state $C$, $\Gamma_\delta$ chooses allocation vector $B_i$ w.p. $w_{B_iC}$ i.e. $P({\bf B^{\Gamma_\delta}[t]} = B_i|C(t) = C) = w_{B_iC} \ \forall \ t$ and decisions across sub-frames are independent. $\delta$ is an input parameter for $\Gamma_\delta$.
\end{mydef}
The definition of $\Gamma_\delta$ will be used for proving various results in this and the following sections.
Consider $\pmb{\lambda} \in \Lambda^\circ$. By the definition of $\Lambda^\circ$, this means that, there exists $\delta > 0$ such that $LP(\delta)$ is feasible for arrival rate vector $\pmb{\lambda}$. Let ${\bf w} = \{w_{B_iC}\}$ denote a feasible solution of $LP(\delta)$.
 Therefore, we can use policy $\Gamma_\delta$ to take scheduling decisions in each sub-frame according to ${\bf w}$. Let $A_k[t]$ denote the arrival process of queue $k$. $A_k[t] = 1$ if there is an arrival to queue $k$ in sub-frame $t$ and $0$ otherwise. $D^{\Gamma_\delta}_k[t]$ denotes the departure process of queue $k$ under $\Gamma_\delta$. $D^{\Gamma_\delta}_k[t]$ is $1$ if a token departs from queue $k$ under $\Gamma_\delta$ in sub-frame $t$ and $0$ otherwise. We have:
 \begin{equation*}
  Q^{\Gamma_\delta}_k[t+1] = \max \{(Q^{\Gamma_\delta}_k[t]+A_k[t]-D^{\Gamma_\delta}_k[t]),0\},
 \end{equation*}
 where $Q^{\Gamma_\delta}_k[t]$ denotes the length of the token queue of UE $k$ at time $t$ under policy $\Gamma_\delta$. For the sake of simplicity of notation, we omit the $\Gamma_\delta$ superscript from $Q^{\Gamma_\delta}_k[t]$ and $D^{\Gamma_\delta}_k[t]$ through the rest of this section.
Since a departure from queue $k$ means that UE $k$ was successfully served, the corresponding service rate $\mu_k^{\Gamma_\delta}[t] = 1$ and we can write the above equation as:
 \begin{equation*}
  Q_k[t+1] = \max \{(Q_k[t]+A_k[t]-\mu_k^{\Gamma_\delta}[t]),0\}.
 \end{equation*}

 The state of the queueing system in a sub-frame can be completely defined by the queue lengths of all the token queues in that sub-frame. We denote the state of the system in sub-frame $t$ by the vector ${\bf Q[t]} = [Q_1[t],\ldots,Q_M[t]]$. 
  Since the scheduling decisions taken under $\Gamma_\delta$ only make use of the current state of the system, the evolution of states of the system $\{\bf{Q[t]}\}_{t\geq0}$ under $\Gamma_\delta$ forms a Discrete Time Markov Chain (DTMC). This DTMC is countable, irreducible and aperiodic. We prove this in the following result.
  
\begin{lemma} \label{lemma:dtmc}
 The DTMC $\{\bf{Q[t]}\}_{t\geq0}$ is countable, irreducible and aperiodic.
\end{lemma}
\begin{proof}
 $\bullet$ \textit{Countable}: The state space of the DTMC is the set of all $M$-tuples $(Q_1[t],\ldots,Q_M[t])$ where $Q_k[t] \in \mathbb{N}$. It forms an $M$ dimensional Cartesian product of the set of natural numbers $\mathbb{N}$ which is a countable set.  Therefore, the state space of the DTMC and hence the DTMC itself is countable (by Theorem $2.13$ in \cite{rudin}). \\
$\bullet$ \textit{Irreducible}: The DTMC can transition from any state $\bf{Q}$ to a state $\bf{Q'}$ in the following steps: \\
$-$ Step 1: Schedule all UEs for service until all queues are empty. This is accomplished in $\max_k Q_k$ sub-frames. \\
$-$ Step 2: For the next $\max_k Q'_k$ sub-frames, the token queue of UE $k$ has an arrival and no departure for the first $Q'_k$ sub-frames. In the remaining $(\max_k Q'_k - Q'_k)$ sub-frames, there is no new arrival and no departure. At the end of this step, the DTMC is in state ${\bf Q'}$.\\
These steps define at least one path of length $(\max_k Q_k + \max_k Q'_k)$ from any state ${\bf Q}$ to any other state ${\bf Q'}$. Therefore, the DTMC is irreducible. \\
$\bullet$ \textit{Aperiodic}: If the DTMC is in state $\bf{Q[t]}$ and no new token arrives in any queue and no queue is scheduled for service, the state of the DTMC remains unchanged. Therefore, self loops exist and the DTMC is aperiodic.
\end{proof}
We now begin the proof of Theorem~\ref{theorem:relation}.


\begin{proof}
We prove Theorem~\ref{theorem:relation} in two steps. We first establish that $\Lambda^\circ \subseteq {\cal S}$ in Lemma~\ref{lemma:relation1} and then show that ${\cal S} \subseteq \overline{\Lambda}$ in Lemma~\ref{lemma:relation2}. 

\begin{lemma} \label{lemma:relation1}
 Every $\pmb{\lambda} \in \Lambda^\circ$ is a stabilizable arrival rate vector. Hence, $\Lambda^\circ \subseteq {\cal S}$.
\end{lemma}
\begin{proof}
 To prove this result, we first show using Foster's theorem~\cite{foster} that the DTMC $\{\bf{Q[t]}\}_{t\geq0}$ is 
 positive recurrent and hence the queue lengths do not grow infinitely under $\Gamma_\delta$. \par
 
Using the Lyapunov function $f({\bf Q[t]}) = \sum_{k=1}^{M} Q_k^2[t]$, we have:
\begin{align*}
\lefteqn{ f({\bf Q[t+1]})-f({\bf Q[t]})} \\
&\leq \sum_{k=1}^M [(A_k(t)-\mu_k^{\Gamma_\delta}[t])^2 + 2Q_k[t](A_k[t] - \mu_k^{\Gamma_\delta}[t])].
\end{align*}
Hence,
\begin{eqnarray}
 \lefteqn{ \mathbb{E}[(f({\bf Q[t+1]})-f({\bf Q[t]}))|{\bf Q[t]}] } \nonumber \\
 &\hspace{-0.6cm}\leq \mathbb{E}\left[(\sum_{k=1}^M [(A_k(t)-\mu_k^{\Gamma_\delta}[t])^2 + 2Q_k[t](A_k[t] - \mu_k^{\Gamma_\delta}[t])])|{\bf Q[t]}\right], \nonumber \\
 &\hspace{-1cm}\leq M + 2\mathbb{E}\left[(\sum_{k=1}^M Q_k[t]A_k[t] - \sum_{k=1}^M Q_k[t]\mu_k^{\Gamma_\delta}[t])|{\bf Q[t]} \right], \nonumber \\
 \label{eqn:1} &\hspace{-1.5cm}\leq M + 2\sum_{k=1}^M Q_k[t]\lambda_k - 2\sum_{k=1}^M Q_k[t]\mathbb{E}\left[\mu_k^{\Gamma_\delta}[t]|{\bf Q[t]} \right]. \nonumber
\end{eqnarray}

From $LP(\delta)$, we have $\mathbb{E}\left[\mu_k^{\Gamma_\delta}[t]|{\bf Q[t]} \right] = \lambda_k+\delta$. Therefore,
\begin{eqnarray*}
 \lefteqn{ \mathbb{E}[(f({\bf Q[t+1]})-f({\bf Q[t]}))|{\bf Q[t]}] } \\
 &\leq M + 2\sum_{k=1}^M Q_k[t]\lambda_k - 2\sum_{k=1}^M Q_k[t](\lambda_k+\delta), \\
 &\leq M - 2\sum_{k=1}^M Q_k[t]\delta.
\end{eqnarray*}
Defining set ${\cal A} = \{{\bf Q} : \sum_{k=1}^M Q_k \leq \frac{M+1}{2\delta}\}$, we have:

\begin{align*}
  \mathbb{E}[(f({\bf Q[t+1]})-f({\bf Q[t]}))|{\bf Q[t]}] < \begin{cases}
          {-1}, & \forall \ {\bf Q[t]} \notin \cal{A},\\
	{\infty}, &  \text{otherwise}.
         \end{cases}
\end{align*}
Thus, by Foster's theorem~\cite{foster}, the DTMC is positive recurrent so the expected queue lengths in the queueing system are finite. Therefore, $\Gamma_\delta$ stabilizes the system for arrival rate vector $\pmb{\lambda} \in \Lambda^\circ$. Thus,  $\pmb{\lambda} \in {\cal S}$ which implies that $\Lambda^\circ \subseteq {\cal S}$. 
\end{proof}
This proves the first part of our result. We now need to show that ${\cal S} \subseteq \overline{\Lambda}$.
In the interest of simplicity of notation, we assume that under a policy $\Gamma$ that stabilizes the system, the following limit exists w.p. $1$.
\begin{equation} \label{eq:limit}
 \lim_{T \rightarrow \infty} \frac{1}{T} \sum_{t=1}^T \mathds{1}_{B_iC}^\Gamma[t],
\end{equation}
where $\mathds{1}_{B_iC}^\Gamma[t]$ is an indicator random variable that is $1$ if allocation vector $B_i$ is chosen by $\Gamma$ under channel state $C$ in sub-frame $t$ and zero otherwise. Now consider the following sets of sample paths:\\
$\bullet$ $A_1$ : the set of sample paths on which Strong Law of Large Numbers (SLLN) holds for the arrival rates i.e. $\frac{\sum_{i=1}^t \lambda_k[t]}{t} \to \lambda_k$ as $ t\rightarrow \infty, \ \forall \ k$. This is a probability $1$ set i.e. $P(A_1) = 1$. \\
$\bullet$ $A_2$ : set of sample paths on which $\frac{\sum_{i=1}^t \mathds{1}_{\{C(t) = C\}}}{t} \to g(C)$ as $ t\rightarrow \infty, \ \forall \ C$ (SLLN holds) where $\mathds{1}_{\{C(t) = C\}}$ is an indicator random variable that is $1$ if the channel state in sub-frame $t$ is $C$ and $0$ otherwise. Since $g$ is a probability distribution over the set of channel states $\cal{C}$, we have, $P(A_2) = 1$. \\
$\bullet$ $A_3$ : the set of sample paths on which service rate under $\Gamma$ is $\geq \pmb{\lambda}$. Since $\Gamma$ stabilizes the system, we have $P(A_3) = 1$. \\
$\bullet$ $A_4$ : the set of sample paths over which the limit in (\ref{eq:limit}) exists. Since we assume that this limit exists w.p. $1$, $P(A_4) = 1$. \\
Since $A_1, A_2, A_3, A_4$ are probability $1$ sets, their intersection, 
\begin{equation} \label{eq:non_trivial}
 A = \bigcap_{i=1}^4A_i
\end{equation}
 is also a probability $1$ set. We refer to the sample paths belonging to this set $A$ as \textit{non-trivial sample paths}. \par
We now prove the second part of our result.

\begin{lemma} \label{lemma:relation2}
If $\pmb{\lambda} \notin \overline{\Lambda}$, then $\pmb{\lambda} \notin {\cal S}$. Thus, ${\cal S} \subseteq \overline{\Lambda}$.
\end{lemma}
\begin{proof}
 We prove this result using a contradiction. Let $\pmb{\lambda} \notin \overline{\Lambda}$ be a stabilizable arrival rate vector i.e. $\pmb{\lambda} \in {\cal S}$.
Since $\pmb{ \lambda}$ is a stabilizable arrival rate vector, there exists some allocation policy $\Gamma$ that can stabilize the system for arrival rate $\pmb{\lambda}$. \par
We observe the scheduling decisions taken by this policy $\Gamma$ along a non-trivial sample path from the set $A$ defined in (\ref{eq:non_trivial}). Let $v_{B_iC}$ denote the fraction of time for which $\Gamma$ chooses the allocation vector $B_i$ under channel state $C$ along such a sample path.  Since $\Gamma$ stabilizes the system, the rate of departures must equal the arrival rate in the system. Therefore, we have:
\begin{align*}
\sum_{C \in {\cal C}} \sum_{B_i \in {\cal B}} g(C) v_{B_iC} \pmb{\mu_{B_iC}} = \pmb{\lambda},\\
\text{where} \ v_{B_iC} \geq 0 \ \forall \ B_i \in {\cal B}, \ C \in {\cal C}, \\
\sum_{B_i \in {\cal B}} v_{B_iC} = 1 \ \forall \ C \in{\cal C}.
\end{align*}
This implies that $\pmb{v} = \{v_{B_iC}\}$ is a feasible solution of $LP(\delta)$ and that,
\begin{equation}
 \pmb{\lambda} \in \Lambda(0) \implies \pmb{\lambda} \in \overline{\Lambda},
\end{equation}

which is a contradiction. Therefore, $\pmb{\lambda} \notin \overline{\Lambda}$ is not stabilizable i.e. any stabilizable $\pmb{\lambda}$ must be contained in $\overline{\Lambda}$. Hence, \\ ${\cal S} \subseteq \overline{\Lambda}$.
\end{proof}

From Lemmas~\ref{lemma:relation1} and~\ref{lemma:relation2}, we have, $\Lambda^\circ \subseteq {\cal S} \subseteq \overline{\Lambda}$, which is the required result. This concludes the proof.
\end{proof}

\subsection{Proof of Lemma~\ref{lemma:stable}} \label{proof:stable}
\begin{proof}
We need to show that the loss requirement of a UE is met if and only if its token queue in the queueing system is stable. We first argue that the stability of the queueing system implies that the loss requirements of the UEs are met.
If the queue corresponding to UE $k$ is stable, it means that there exists a policy $\Gamma$ that stabilizes the queue for $\pmb{\lambda} \in \Lambda^\circ$. We can, therefore, construct a randomized policy $\Gamma_\delta$ as defined in Definition~\ref{def:randomized} (Appendix~\ref{proof:relation}). Under $\Gamma_\delta$, the rate of service is greater than $\lambda_k$ which means that UE $k$ is served in greater than $(1-\tilde{\ell}_k)$ of the sub-frames. Therefore, the loss encountered by UE $k$ is less than $\tilde{\ell}_k$ and its loss requirements are successfully met. \par

Now, let us assume that the loss requirement of UE $k$ is met. We show that this ensures the stability of its token queue. Since the loss requirement $\tilde{\ell}_k$ is achievable, there exists a policy $\Gamma$ that satisfies the loss requirement. This means that, under $\Gamma$, the UE is being served in greater than $(1-\tilde{\ell}_k)$ fraction of sub-frames. Since the arrival rate $\lambda_k = (1-\tilde{\ell}_k)$, the queue is served at a rate greater than the arrival rate. Hence, $\Gamma$ stabilizes the token queue. From these arguments, we conclude that the loss requirement of a UE is met iff its corresponding token queue in the queueing system is stable.
Therefore, the feasible region of the optimal allocation policy $\Gamma^\star$, ${\cal L}^{\Gamma^\star}$ is equivalent to the stability region of the queueing system, $\cal{S}$.
\vspace{-0.3cm}
\end{proof}

\subsection{Proof of Theorem~\ref{theorem:BP}} \label{proof:BP}
\begin{proof}
Let $D^{\Gamma_0}_k[t]$ denote the departure process of queue $k$ under $\Gamma_0$. We have:
 \begin{equation*}
  Q^{\Gamma_0}_k[t+1] = \max\{(Q^{\Gamma_0}_k[t]+A_k[t]-D^{\Gamma_0}_k[t]),0\},
 \end{equation*}
 where $Q^{\Gamma_0}_k[t]$ denotes the queue length of the token queue of $k$ at time $t$ under $\Gamma_0$. For the sake of simplicity of notation, we omit the $\Gamma_0$ superscript from $Q^{\Gamma_0}_k[t]$ and $D^{\Gamma_0}_k[t]$ through the rest of this section.
Since a departure from queue $k$ means that UE $k$ was successfully served, the service rate $\mu^{\Gamma_0}_k[t] = 1$ and we can write the above equation as:
 \begin{equation*}
  Q_k[t+1] = \max\{(Q_k[t]+A_k[t]-\mu^{\Gamma_0}_k[t]),0\}.
 \end{equation*}
 The state of the queueing system is completely defined by the vector ${\bf Q[t]} = [Q_1[t], \ldots ,Q_M[t]]$.
 The evolution of $\bf{Q[t]}$ forms a DTMC since the scheduling decisions taken by $\Gamma_0$ in a sub-frame are based solely on the state of the system in that sub-frame. The DTMC is countable, irreducible and aperiodic. The proof that the DTMC has these properties follows the same arguments as in Lemma~\ref{lemma:dtmc} in Appendix~\ref{proof:relation}. We now show using Foster's theorem~\cite{foster} that this DTMC is positive recurrent and hence the token queues do not grow infinitely. \par
 
Using the Lyapunov function $f({\bf Q[t]}) = \sum_{k=1}^{M} Q_k^2[t]$, we have:
\begin{align*}
\lefteqn{ f({\bf Q[t+1]})-f({\bf Q[t]})} \\
&= \sum_{k=1}^M [(A_k(t)-\mu^{\Gamma_0}_k[t])^2 + 2Q_k[t](A_k[t] - \mu^{\Gamma_0}_k[t])].
\end{align*}
Hence,
\begin{align}
 \lefteqn{ \mathbb{E}[(f({\bf Q[t+1]})-f({\bf Q[t]}))|{\bf Q[t]}] } \nonumber \\
 &= \mathbb{E}\left[(\sum_{k=1}^M [(A_k(t)-\mu^{\Gamma_0}_k[t])^2 + 2Q_k[t](A_k[t] - \mu^{\Gamma_0}_k[t])])|{\bf Q[t]}\right], \nonumber \\
&\leq M + 2\sum_{k=1}^M Q_k[t]\lambda_k - 2\mathbb{E}\left[(\sum_{k=1}^M Q_k[t]\mu^{\Gamma_0}_k[t])|{\bf Q[t]} \right]. \label{eqn:1}
\end{align}

Let $\mu^{\Gamma_\delta}_k[t]$ denote the service rate for UE $k$ in sub-frame $t$ under the randomized policy $\Gamma_\delta$. Then, from~(\ref{eqn:mw_objective}), we have:
\begin{equation}
\label{eqn:max} \sum_{k=1}^M Q_k[t]\mu^{\Gamma_0}_k[t] \geq \sum_{k=1}^M Q_k[t]\mu^{\Gamma_\delta}_k[t].
\end{equation}
Therefore, from (\ref{eqn:1}) and (\ref{eqn:max}):
\begin{align*}
 \lefteqn{ \mathbb{E}[(f({\bf Q[t+1]})-f({\bf Q[t]}))|{\bf Q[t]}] } \\
 &\leq M + 2\sum_{k=1}^M Q_k[t]\lambda_k - 2\mathbb{E}\left[(\sum_{k=1}^M Q_k[t]\mu^{\Gamma_\delta}_k[t])|{\bf Q[t]} \right], \\
 &\leq M + 2\sum_{k=1}^M Q_k[t]\lambda_k - 2\sum_{k=1}^M Q_k[t](\lambda_k+\delta), \\
 &\leq M - 2\sum_{k=1}^M Q_k[t]\delta.
\end{align*}
Now consider the set ${\cal A} = \{{\bf Q} : \sum_{k=1}^M Q_k \leq \frac{M+1}{2\delta}\}$, we have:

\begin{align*}
  \mathbb{E}[(f({\bf Q[t+1]})-f({\bf Q[t]}))|{\bf Q[t]}] < \begin{cases}
          {-1}, & \forall \ {\bf Q[t]} \notin \cal{A},\\
	{\infty}, &  \text{otherwise}.
         \end{cases}
\end{align*}
Thus, by Foster's theorem~\cite{foster}, the DTMC is positive recurrent which means that the expected queue lengths in the queueing system will be finite. Therefore, $\Gamma_0$ stabilizes the system and hence meets the loss requirements of the UEs.
\vspace{-0.2cm}
\end{proof}

\subsection{Proof of Lemma~\ref{lemma:dtmc2}} \label{proof:dtmc2}
For the sake of simplicity of notation, we omit the $\Gamma_P$ superscript from the notations through the rest of this section.
\begin{proof}
$\bullet$ \textit{Countable}: The state of the DTMC ${\bf Q[t]}$ is comprised of the queue lengths of the $M$ UEs and their priority weights. We have already shown in Lemma~\ref{lemma:dtmc} that the state space of queue lengths $(Q_1[t], \ldots ,Q_M[t])$ is a countable set. The state space of the priority weights of the UEs is an $M$ dimensional Cartesian product over the finite set $\{1,2, \ldots ,\kappa\}$ and is therefore a finite countable set (by Theorem $2.13$ in \cite{rudin}). Therefore, the states of the DTMC $\bf{Q[t]}$ form a $2M$ dimensional Cartesian product over two countable sets, the state space of queue lengths and the state space of the priority weights. Therefore, the state space of the DTMC and hence the DTMC itself is countable (by Theorem $2.13$ in \cite{rudin}). \\
$\bullet$ \textit{Irreducible}: Consider that the DTMC is in state ${\bf Q} = \{Q_1, \ldots ,Q_M,\bar{c}_k\}$. We will show that a path exists from $\bf{Q}$ to any state ${\bf Q'} = \{Q'_1, \ldots ,Q'_M,\bar{c'}_k\}$. The DTMC can transition from $\bf{Q}$ to $\bf{Q'}$ in the following steps:\\
$-$ Step 1: Schedule all UEs for service until all queues are empty. This is accomplished in $\max_k Q_k$ sub-frames.\\
 $-$ Step 2: A new token arrives in every queue and no queue is scheduled for service for the next $\min_k Q'_k$ sub-frames. At the end of this step, all queue lengths are equal to $\min_k Q'_k$ and all priority weights are equal to $\min(\min_k Q'_k,\kappa)$. \\
$-$ Step 3: For the next $\max_k Q'_k - \min_k Q'_k$ sub-frames, the UEs in $\arg \max_k Q'_k$ see one arrival and no departure. Every other UE $k'$ see an arrival and no departure for the first $(Q'_{k'}- \min_k Q'_k)$ sub-frames and one arrival and one departure for the remaining $\max_k Q'_k - Q'_{k'}$ sub-frames. At the end of this step, the queue length of UE $k$ is equal to $Q'_k$.\\
$-$ Step 4: In the next sub-frame, there is one arrival and one departure in every queue. This makes the priority weights all equal to $0$ while the queue lengths remain unchanged.\\
$-$ Step 5: In the next $\max_k c'_k$ sub-frames, there is no arrival and no departure for UEs in $\arg \max_k c'_k$. For every other UE $k'$, there is an arrival and a departure in the first $(\max_k c'_k - c'_{k'})$ of these sub-frames and no arrival and no departure in the remaining $c'_{k'}$ sub-frames. At the end of this step, the DTMC is in the desired state $\bf{Q'}$.\\
This defines one finite length path from any state $\bf{Q}$ to any other state ${\bf Q'}$ of length $(\max_k Q_k + \max_k Q'_k +1+\max_k c'_k)$. Hence, the DTMC is irreducible. \\
$\bullet$ \textit{Aperiodic}: Consider state $\bf{Q[t]}$ where all queues are empty and all priority weights are $0$. If there is one arrival in each queue in slot $t+1$ and every queue is served, the queues remain empty and the priority weights remain $0$. Therefore, this state has a self loop and hence has period $1$. Since we have already shown that the DTMC is irreducible, all states have period $1$ because periodicity is a class property. Hence, the DTMC is aperiodic. 
\vspace{-0.7cm}
\end{proof}

\subsection{Proof of Theorem~\ref{theorem:BP2}} \label{proof:BP2}
\begin{proof}
Let $D^{\Gamma_P}_k[t]$ denote the departure process of queue $k$ under $\Gamma_P$. We have:
 \begin{equation*}
  Q^{\Gamma_P}_k[t+1] = \max\{(Q^{\Gamma_P}_k[t]+A_k[t]-D^{\Gamma_P}_k[t]),0\},
 \end{equation*}
 where $Q^{\Gamma_P}_k[t]$ is the queue length of the token queue of $k$ at time $t$ under $\Gamma_P$. For the sake of simplicity of notation, we omit the $\Gamma_P$ superscript from $Q^{\Gamma_P}_k[t]$ and $D^{\Gamma_P}_k[t]$ through the rest of this section. Since a departure from queue $k$ in sub-frame $t$ means that $\mu^{\Gamma_P}_k[t] = 1$, we can write the above equation as:
 \begin{equation*}
  Q_k[t+1] = \max\{(Q_k[t]+A_k[t]-\mu^{\Gamma_P}_k[t]),0\}.
 \end{equation*}
 Under this policy, the evolution of the state of the queueing system ${\bf Q[t]} = [Q_1[t],\ldots,Q_M[t],\bar{c}[t]]$ forms a DTMC. We have proved in Lemma~\ref{lemma:dtmc2} that this DTMC is countable, irreducible and aperiodic. We now show using Foster's theorem~\cite{foster} that this DTMC is positive recurrent and hence the queues do not grow infinitely. \par
 
Using the following Lyapunov function $f({\bf Q[t]}) = \sum_{k=1}^{M} Q_k^2[t]$, we have:
\begin{align*}
\lefteqn{ f({\bf Q[t+1]})-f({\bf Q[t]})} \\
&= \sum_{k=1}^M [(A_k(t)-\mu^{\Gamma_P}_k[t])^2 + 2Q_k[t](A_k[t] - \mu^{\Gamma_P}_k[t])].
\end{align*}
Hence, as in (\ref{eqn:1}), we have:
\begin{align}
 \lefteqn{ \mathbb{E}[(f({\bf Q[t+1]})-f({\bf Q[t]}))|{\bf Q[t]}] } \nonumber \\
 \label{eqn:2} &\leq M + 2\sum_{k=1}^M Q_k[t]\lambda_k - 2\mathbb{E}\left[(\sum_{k=1}^M Q_k[t]\mu^{\Gamma_P}_k[t])|{\bf Q[t]} \right].
\end{align}
Let $\mu^{\Gamma_\delta}_k[t]$ denote the service rate for UE $k$ in sub-frame $t$ under the randomized policy $\Gamma_\delta$. Then, from~(\ref{eq:objective_priority}), we have:
\begin{align}
\label{eqn:max_priority} \lefteqn{ \sum_{k=1}^M \left(Q_k[t]\mu^{\Gamma_P}_k[t] +(c_k[t]+1)s\mu^{\Gamma_P}_k[t] \right)} \nonumber \\
&\geq \sum_{k=1}^M \left(Q_k[t]\mu^{\Gamma_\delta}_k[t]+(c_k[t]+1)s\mu^{\Gamma_\delta}_k[t] \right).
\end{align}
Therefore, from (\ref{eqn:2}) and (\ref{eqn:max_priority}):
\begin{eqnarray*}
 &\hspace{-0.6cm} \mathbb{E}[(f({\bf Q[t+1]})-f({\bf Q[t]}))|{\bf Q[t]}] \leq M + 2\sum_{k=1}^M Q_k[t]\lambda_k \\
 &\hspace{-0.6cm}- 2\mathbb{E}\left[(\sum_{k=1}^M Q_k[t]\mu^{\Gamma_\delta}_k[t]+(c_k[t]+1)(\mu^{\Gamma_\delta}_k[t]-\mu^{\Gamma_P}_k[t])s)|{\bf Q[t]} \right], \\
 &\hspace{-0.6cm}\leq M + 2\sum_{k=1}^M Q_k[t]\lambda_k - 2\sum_{k=1}^M Q_k[t](\lambda_k+\delta) \\
 &- 2\mathbb{E}\left[(\sum_{k=1}^M -(\kappa+1)s\right], \\
 &\leq M - 2\sum_{k=1}^M Q_k[t]\delta +4Ms. \ (\text{for} \ \kappa = 1)
\end{eqnarray*}


Defining set ${\cal A} = \{{\bf Q} : \sum_{k=1}^M Q_k \leq \frac{4Ms+M+1}{2\delta}\}$, we have:
\begin{align*}
  \mathbb{E}[(f({\bf Q[t+1]})-f({\bf Q[t]}))|{\bf Q[t]}] < \begin{cases}
          {-1}, & \forall \ {\bf Q[t]} \notin \cal{A},\\
	{\infty}, &  \text{otherwise}.
         \end{cases}
\end{align*}
Thus, by Foster's theorem~\cite{foster}, the DTMC is positive recurrent which means that the expected queue lengths in the queueing system will be finite. Therefore, $\Gamma_P$ stabilizes the system and hence meets the loss requirements of all the UEs.
\end{proof}
\vspace{-0.3cm}

\subsection{Proof of Lemma~\ref{lemma:mwbm}} \label{proof:mwbm}
\begin{proof}
 A matching for graph ${\cal G}$ selects edges that share no common vertices. This means that each group from $U$ will be matched to exactly one PRB from $V$ and each PRB from $V$ will be matched to at most one group from $U$. Therefore, the requirement of
 assigning no more than $1$ PRB to each group is satisfied. Since PRBs in $V$ are matched to no more than one group from $U$, we will have $B^\Gamma_i[t] \neq B^\Gamma_{i'}[t] \ \forall \ \{i,i' \in [L] : B^\Gamma_i[t], B^\Gamma_{i'}[t] \neq 0\}$ as
 required by Definition~\ref{def:RA}. Thus, the solution of the MWBM gives us a feasible resource allocation. Next, we show that the resulting allocation is consistent with the allocation decisions that would be taken by policy $\Gamma_0$. \par
 
 MWBM picks edges such that the sum of the weights of the edges chosen is maximized. Therefore, it maximizes the quantity $\sum_{i \in U}\sum_{k \in G_i} Q_k[t]\mu_k^j[t] = \sum_{k=1}^{M} Q_{k}[t] \mu_{k}[t]$ which is same as in~(\ref{eqn:mw_objective}). 
 Hence, resource allocation done using MWBM on ${\cal G}$ is consistent with policy $\Gamma_0$.
 \vspace{-0.5cm}
 \end{proof}

\bibliographystyle{ieeetr}
\addcontentsline{toc}{chapter}{Bibliography}
\bibliography{backpressure}


\end{document}